\documentclass{article}
\usepackage{graphicx} 
\usepackage{color,xcolor}
\usepackage{epsf, graphicx}
\usepackage{latexsym,amsfonts,amsbsy,amssymb}
\usepackage{amsmath,amsthm}
\usepackage{comment}
\usepackage{booktabs}
\usepackage{cleveref}
\usepackage{indentfirst}
\usepackage{bm}
\usepackage{url}
\usepackage [latin1]{inputenc}
\makeatletter

\@addtoreset{equation}{section} \makeatother

\newtheorem{Lemma}{Lemma}
\newtheorem{Corollary}{Corollary}
\newtheorem{Remark}{Remark}
\newtheorem{Definition}{Definition}
\newtheorem{Proposition}{Proposition}
\newtheorem{Example}{Example}
\makeatletter \@addtoreset{equation}{section} \makeatother
\title{About the Rankin and Berg\'e-Martinet Constants from a Coding Theory View Point}
\author{Fr\'{e}d\'{e}rique Oggier, Shengwei Liu,~Hongwei Liu }
\date{}

\newcommand{\FF}{\mathbb{F}}
\newcommand{\RR}{\mathbb{R}}
\newcommand{\ZZ}{\mathbb{Z}}

\begin{document}

\maketitle
\insert\footins{\small
\noindent
Fr\'{e}d\'{e}rique Oggier (Corresponding author) is currently with the School of Mathematics, the University of Birmingham, UK, B15 2TT
 (e-mail:
f.e.oggier@bham.ac.uk). Part of this work was done at Nanyang Technological University, Singapore.\\
Shengwei Liu is with the School of Mathematics and  Statistics, Central China Normal University, Wuhan 430079, China (e-mail: shengweiliu@mails.ccnu.edu.cn).\\
Hongwei Liu is with the School of Mathematics and Statistics and the Key Laboratory NAA-MOE, Central China Normal University, Wuhan 430079,
China (e-mail: hwliu@ccnu.edu.cn).\\
 }

\begin{abstract}
The Rankin constant $\gamma_{n,l}$  measures the largest volume of the densest sublattice of rank $l$ of a lattice $\Lambda\in \RR^n$ over all such lattices of rank $n$. The Berg\'e-Martinet constant $\gamma'_{n,l}$ is a variation that takes into account the dual lattice. Exact values and bounds for both constants are mostly open in general.
We consider the case of lattices built from linear codes, and look at bounds on $\gamma_{n,l}$ and $\gamma'_{n,l}$. In particular, we revisit known results for $n=3,4,5,8$ and give lower and upper bounds for the open cases $\gamma_{5,2},\gamma_{7,2}$ and $\gamma'_{5,2},\gamma'_{7,2}$.
\end{abstract}

%
%

\section{Introduction}

Let $\Lambda$ be a lattice of rank $n$ in $\RR^n$. A well studied question is to determine the length of shortest non-zero vectors in $\Lambda$, that is \cite[Def. 1.2.1]{martinet}
\begin{equation}\label{def:minnorm}
N(\Lambda) = \min_{x\in \Lambda,~x\neq 0}||{\bm x}||^2.
\end{equation}
When comparing two lattices, scaling effects should not be play a role and the Hermite invariant of $\Lambda$ is defined as \cite[Def. 2.2.5]{martinet}
\[
\gamma_n(\Lambda) = \frac{N(\Lambda)}{\det(\Lambda)^{1/n}}
\]
where $\det(\Lambda)$ is defined as the determinant of a Gram matrix of $\Lambda$, and a Gram matrix $\bm G$ is defined by ${\bm G}={\bm B}{\bm B}^T$ for ${\bm B}$ a generator matrix, whose rows form a basis of $\Lambda$.
Finding a lattice of rank $n$ in $\RR^n$ with the largest shortest non-zero vector leads to the classical Hermite constant \cite[Def. 2.2.6]{martinet}:
\[
\gamma_n=\sup_{rank(\Lambda)=n}\gamma_n(\Lambda).
\]
Determining $\gamma_n$ remains mostly open, since its values are only known \cite{conj} for $n\in\{1,2,3,4,5,6,7,8,24\}$.

For $\Lambda$ a lattice of rank $n$, and $\Lambda_l\leq\Lambda$ a sublattice of rank $l$, set
\begin{equation}\label{eq:gammaLLk}
\gamma(\Lambda,\Lambda_l) = \frac{\det(\Lambda_l)}{\det(\Lambda)^{l/n}},
\end{equation}
which generalizes the case $l=1$ for which $\Lambda_l$ is generated by a non-zero vector. Set
\begin{equation}\label{eq:dl}
d_l(\Lambda)=\min_{
\substack{
\Lambda'\leq \Lambda, \\
rank(\Lambda')=l}} {\det(\Lambda')},~1\leq l \leq n,
\end{equation}
and define \cite[Def. 2.8.3]{martinet}
\begin{eqnarray}
    \gamma_{n,l}(\Lambda) &=& d_l(\Lambda)\det(\Lambda)^{-l/n}  \label{eq:gammanlLambda}\\
    \gamma_{n,l} &=& \sup_{rank(\Lambda)=n}\gamma_{n,l}(\Lambda) \label{eq:gammanl}\\
     \gamma'_{n,l}(\Lambda) &=& \sqrt{d_l(\Lambda)d_l(\Lambda^*)}  \label{eq:gammaprimenlLambda}\\
    \gamma'_{n,l} &=& \sup_{rank(\Lambda)=n}\gamma'_{n,l}(\Lambda) \label{eq:gammaprimenl},
\end{eqnarray}
where $\Lambda^\ast$ denotes the dual lattice of $\Lambda$. Note that $\gamma(\Lambda,\Lambda_l)$ is an upper bound on $\gamma_{n,l}(\Lambda)$. Then $\gamma_{n,l}$ is called the Rankin constant, and when $l=1$, it translates into the Hermite constant: $\gamma_{n,1}=\gamma_n$. The constant $\gamma'_{n,l}$ is sometimes called the Berg\'e-Martinet constant. Determining $\gamma_{n,l}$ and  $\gamma'_{n,l}$ also remains mostly open. The known values of $\gamma_{n,l}$ and $\gamma'_{n,l}$   (see \cite{Sawatani})
are reported in Table  \ref{tab:my_label}. Since \cite[Prop. 2.8.5]{martinet}
\begin{eqnarray*}
\gamma_{n,l} &=& \gamma_{n,n-l} \\
\gamma'_{n,l} &=& \gamma'_{n,n-l},
\end{eqnarray*}
only values for $l< \lfloor n/2 \rfloor$ are shown.

\begin{table}
    \centering
    \begin{tabular}{cccccc}
    \toprule
        $n$ & $l$ & $\gamma_{n,l}$ & $\Lambda$ & $\gamma^{'}_{n,l}$ & $\Lambda$  \\
    \midrule
        2 & 1 & $\frac{2}{\sqrt{3}}$ & $A_2$ & $\frac{2}{\sqrt{3}}$ & \\
    \midrule
        3 & 1 & $\sqrt[3]{2}$ & $A_3\simeq D_3$ $\checkmark$ & $\sqrt{3/2}$ & $D_3$ $\checkmark$ \\
    \midrule
        4 & 1 & $\sqrt{2}$ & $D_4$ $\checkmark$  & $\sqrt{2}$ & $D_4$ $\checkmark$ \\
        4 & 2 & 3/2  & $D_4$ $\checkmark$ & 3/2 &  $D_4$ $\checkmark$  \\
    \midrule
        5 & 1 & $\sqrt[5]{8}$ & $D_5$  $\checkmark$  & $\sqrt{2}$ & $D_5$  $\checkmark$  \\
        5 & 2 &   &   \\
    \midrule
        6 & 1 & $(\frac{64}{3})^{1/6}$  &  $E_6$ & $\sqrt{8/3}$  \\
        6 & 2 & $3^{2/3}$ & $E_6$ & 2 &$E_6$ \\
        6 & 3 &  &  \\
    \midrule
        7 & 1 & $\sqrt[7]{64}$ &  $E_7$ & $\sqrt{3}$  \\
        7 & 2 & & \\
        7 & 3 & & \\
    \midrule
        8 & 1 & 2  & $E_8$ $\checkmark$ & 2 &  $E_8$ $\checkmark$  \\
        8 & 2 & 3  & $E_8$ $\checkmark$ &3 & $E_8$ $\checkmark$  \\
        8 & 3 & 4 & $E_8$ &4 &$E_8$ \\
        8 & 4 & 4 & $E_8$ &4 &$E_8$ \\
    \bottomrule
    \end{tabular}
     \caption{Known values of $\gamma_{n,l}$ and $\gamma'_{n,l}$, with on their right a lattice achieving them. A check mark indicates the existence of a construction by code, as presented in this paper.}
    \label{tab:my_label}
\end{table}

The following other equalities and inequalities (see \cite[Section 2.8]{martinet}) are known.
Let $\Lambda$ be a lattice, for $0<l<n$, we have:
\begin{enumerate}
\item
$\gamma_{n,l}(\Lambda)\leq\gamma_{n,1}(\Lambda)^{l}$ and $\gamma^{'}_{n,l}(\Lambda)^{2}=\gamma_{n,l}(\Lambda)\gamma_{n,l}(\Lambda^{\ast})$.
\item
$\gamma^{'}_{n,l}\leq\gamma_{n,l}\leq(\gamma_{n})^{l}$.
\item
$\gamma_{n,l}(\Lambda)=\gamma_{n,n-l}(\Lambda^{\ast})$ and $\gamma^{'}_{n,l}(\Lambda)=\gamma^{'}_{n,n-l}(\Lambda)=\gamma^{'}_{n,l}(\Lambda^{\ast})=\gamma^{'}_{n,n-l}(\Lambda^{\ast})$.
\item
For $0\leq l\leq h\leq n$, we have $\gamma_{n,l}\leq\gamma_{h,l}(\gamma_{n,h})^{l/h}$.
\item
For $0\leq l\leq n/2$, we have $(\gamma_{n,l})^{n}\leq(\gamma_{n-l,l})^{n-l}(\gamma^{'}_{n,l})^{2l}$ and $\gamma^{'}_{n,2l}\leq(\gamma^{'}_{n-l,l})^{2}$.
\item
If $n$ is even, then $\gamma^{'}_{n,n/2}=\gamma_{n,n/2}$.
\item
For $0\leq l\leq n$, we have $(\gamma_{n,l})^{n-2l}\leq(\gamma_{n-l,l})^{n-l}$.
\item
For a positive $l$, $\gamma^{'}_{2l+1}\leq(\gamma^{'}_{l+1})^{2}$.
\end{enumerate}

As for bounds, it is known \cite{crypto} that
\[
\left( \frac{k}{12} \right)^{k/2}
\leq \gamma_{2k,k}
\leq
\left( 1+\frac{k}{2} \right)^{k\ln 2 +1/2},~k\geq 2,
\]
a result that was improved to
\[
\frac{4}{\pi^2\sqrt{k}}\left(\frac{2k}{\pi e^{3/2}} \right)^{k/2}
\leq
\gamma_{2k,k}
\leq
e^9
(0.0833)^{k/2}
\left(
\frac{4k-1}{17}
\right)^{\tfrac{k}{4k-2}}
(k-0.5)^{k\ln 2},~k \geq 5,
\]
in \cite{Wen}. Bounds for arbitrarily large dimensions are of interest for lattice reductions, a topic that attracted a renewed interest in the context of lattice-based cryptography \cite{crypto,Wen}.

Let $C$ be a linear code in $\mathbb{Z}_q^n$, where $\mathbb{Z}_q$ denotes the integer modulo $q$ for $q$ an integer. When $q=p$ is a prime, we will write $\ZZ_p = \FF_p$ to distinguish the field structure from the ring structure of $\ZZ_q$.
A well known technique (see e.g.~\cite{gaborit}) to construct lattices with a large shortest vector is to start from a linear code $C$, and to construct a corresponding lattice $\Lambda_C$ through the so-called Construction A (see e.g.~\cite[Lattices from Codes]{costa}). The Hamming weight of codewords in $C$, defined as the number of non-zero coefficients of a codeword, is used to bound the norm of vectors in $\Lambda_C$. Variations by considering codes over other alphabets are also available, see e.g. \cite{Betti,Freed} for  constructions over $\FF_4$ and $\FF_2 \times \FF_2$, which furthermore illustrate a deeper connection between codes and lattices, namely between their respective weight enumerators and theta series.

The question that we address here is to extend the above technique to find sublattices of rank $1\leq l \leq k$ of a lattice $\Lambda_C$ such that $\gamma_{n,l}(\Lambda_C)$ is large. In Section \ref{sec:lfc}, we provide some background on lattices from codes, before proving general bounds on $\gamma_{n,l}$ and $\gamma'_{n,l}$, which only depend on the use of Construction A. In Sections \ref{sec:gamma} and \ref{sec:gamma'}, we look at specific codes and their corresponding lattices to provide results on $\gamma_{n,l}$ and $\gamma'_{n,l}$ respectively. In particular, we revisit known results for $n=3,4,5,8$, and give lower and upper bounds for the open cases $\gamma_{5,2},\gamma_{7,2}$ and $\gamma'_{5,2},\gamma'_{7,2}$.

%
%
%
\section{Lattices from Codes}
\label{sec:lfc}

Let $q\geq2$ be a positive integer and $\mathbb{Z}_{q}$ be the integers modulo $q$. 
A linear code $C$ in $\mathbb{Z}_{q}^{n}$ is an additive subgroup of $\mathbb{Z}_{q}^{n}$. If $q$ is a prime, then $C$ is a linear subspace of $\FF_q^n$.

Define the map $\rho$ from $\mathbb{Z}^{n}$ to $\mathbb{Z}_{q}^{n}$ as:
$$\rho:\mathbb{Z}^{n}\rightarrow\mathbb{Z}_{q}^{n},~(x_{1},\dots,x_{n})\mapsto(x_{1}\!\!\!\!\pmod q,\dots,x_{n}\!\!\!\!\pmod q).$$

\begin{Definition}
Let $C$ be a linear code in $\mathbb{Z}_{q}^{n}$, then the lattice $\Lambda_{C}=\rho^{-1}(C)$ is said to be obtained via Construction A.
\end{Definition}

The lattice $\Lambda_C$ is integral, meaning that ${\bm G}$ as integer coefficients.

\begin{Lemma}\label{lem:qZn}
We have $q\ZZ^n \subseteq \Lambda_C$ for $C$ a linear code.
\end{Lemma}
\begin{proof}
Indeed $C$ is linear thus contains ${\bf 0}$ and $\rho(q\ZZ^n)={\bf 0}$ so $ q\ZZ^n \subseteq \rho^{-1}(C)$.
\end{proof}

\begin{Proposition}\label{prop:detlambdaC}
We have
\[
\det(\Lambda_C)=\left(\frac{q^n}{|C|}\right)^2.
\]
\end{Proposition}
\begin{proof}
Since $\Lambda_C/q\ZZ^n \simeq C$ by the first isomorphism theorem for groups, we have $|\Lambda_C/q\ZZ^n| = |C|$. Then a generator matrix of the sublattice $q\ZZ^n$ is ${\bm B}'={\bm A}{\bm B}$ for ${\bm B}$ a generator matrix of $\Lambda_C$ and ${\bm A}$ an integral matrix, so
\[
|C| = |\Lambda_C/q\ZZ^n| = |\det({\bm A}) | = \frac{|\det({\bm B}')|}{|\det({\bm B})|}
\]
and
\[
\det(\Lambda_C)=
\det({\bm B})^2
=
\left(\frac{q^n}{|C|}\right)^2.
\]
\end{proof}

\begin{Corollary}\label{cor:gammanlforLambdaC}
For $C$ a linear code in $\ZZ_q^n$, we have
\[
\gamma_{n,l}(\Lambda_C) = d_l(\Lambda_C)\frac{|C|^{2l/n}}{q^{2l}}.
\]
\end{Corollary}
\begin{proof}
From (\ref{eq:gammanlLambda}) and Proposition \ref{prop:detlambdaC}:
\[
\gamma_{n,l}(\Lambda_C)
= \frac{d_l(\Lambda_C)}{\det(\Lambda_C)^{l/n}}
= \frac{d_l(\Lambda_C)}{(q^n/|C|)^{2l/n}}
=  \frac{d_l(\Lambda_C)}{q^{2l}/|C|^{2l/n}}
=  d_l(\Lambda_C)\frac{|C|^{2l/n}}{q^{2l}}.
\]
\end{proof}

We thus get a first upper bound on $\gamma_{n,l}(\Lambda_C)$, which will be shown to be tight for $l=1$ in Remark \ref{rem:l1tight} later on.

\begin{Corollary}\label{cor:ub}
We have $d_l(\Lambda_C)\leq q^{2l}$ so
\[
\gamma_{n,l}(\Lambda_C) \leq |C|^{2l/n}.
\]
\end{Corollary}
\begin{proof}
From (\ref{eq:dl}),
\[
d_l(\Lambda_C)\leq q^{2l}
\]
since by Lemma \ref{lem:qZn} the lattice $q\ZZ^n$ is always a sublattice of $\Lambda_C$, which always contains the sublattice $q\ZZ^l$ for $1\leq l \leq n$, for which a Gram matrix is $q^2I_l$ and $\det(q\ZZ^l)=q^{2l}$.

Then apply Corollary \ref{cor:gammanlforLambdaC}.
\end{proof}

The Hamming weight $w_H$ of a codeword of $C$ is the number of its non-zero
coefficients and $d_H(C)$ is the minimum Hamming weight over all non-zero codewords of $C$. The Euclidean weight $w_E(a)$ of an element $a \in \ZZ_q$
is $\min (a^2, (q-a)^2)$. This is because $\rho^{-1}(a)=a+bq$ for some integer $b$, so if $b \geq 0$, $(a+bq)^2 \geq a^2$, and if $b <0$, $(a+bq)^2=(-a-bq)^2\geq (-a+q)^2$.
By extension, the Euclidean weight of a codeword ${\bm c}=(c_1,\ldots,c_n)$ is
\[
w_E({\bm c}) =\sum_{i=1}^n w_E(c_i)
\]
which is the smallest norm among the $|| {\bm x} ||^2$ for $x\in \rho^{-1}({\bm c})$. The Euclidean minimum weight $d_E(C)$ of a code $C$ is thus the minimum
Euclidean weight among all the non-zero codewords of $C$.

\begin{Corollary}\label{cor:l1}
For $l=1$, we have $d_1(\Lambda_C) = \min(q^{2},d_E(C))$, so
\[
\gamma_{n,1}(\Lambda_C) =  \min(q^{2},d_E(C))\frac{|C|^{2/n}}{q^{2}}.
\]
\end{Corollary}
\begin{proof}
A lattice point comes either from a lift of ${\bf 0}$, which using Corollary \ref{cor:ub} gives a norm of $q^2$, or from a non-zero codeword of $C$. Then apply Corollary \ref{cor:gammanlforLambdaC}.
\end{proof}

\begin{Proposition}\label{prop:ubl2}
We have for $l=2$
\[
d_2(\Lambda_C) \leq \min(q^{4},q^2(d_E(C)-b^2)) \leq
\min(q^{4},q^2d_E(C)),
\]
where $b$ is a coefficient of ${\bm b}$, a lattice point achieving the smallest Euclidean weight $d_E(C)$, such that $|b|$ is maximal among $b_1,\ldots,b_n$ and among choices of ${\bm b}$. Then
\[
\gamma_{n,2}(\Lambda_C) \leq \min (1, \tfrac{1}{q^2}(d_E(C)-b^2))|C|^{4/n}.
\]
\end{Proposition}
\begin{proof}
Suppose $l=2$, and take two lattice vectors ${\bm a},{\bm b}$, which are linearly independent, so they generate a sublattice $\Lambda'$ of $\Lambda_C$ of rank 2. Each such lattice vector is the preimage of $\rho^{-1}({\bm c})$ for ${\bm c}$ some codeword of $C$, which is either ${\bf 0}$ or not.
If both are in the preimage of ${\bf 0}$, we already know from Corollary \ref{cor:ub} that $q^4$ is an upper bound on $d_2(\Lambda_C)$. Suppse now that one of them. say ${\bm a}$, is in the preimage of ${\bf 0}$, say ${\bm a}=(0,\ldots,0,q,0,\ldots,0)$, where $q$ is in the $i$th position, but the other vector ${\bm b}$ is in the preimage of  $\rho^{-1}({\bm c})$ for ${\bm c}$ some non-zero codeword of $C$. Pick ${\bm c}$ such that $w_E({\bm c})$ is minimal, that is $w_E({\bm c})=d_E(C)$.
Then a Gram matrix of $\Lambda'$ is
\[
\begin{bmatrix}
\langle {\bm a},{\bm a} \rangle &  \langle {\bm a},{\bm b}\rangle \\
\langle {\bm a}, {\bm b} \rangle & \langle {\bm b},{\bm b} \rangle
\end{bmatrix}
=\begin{bmatrix}
q^2 & qb_i \\
qb_i  & d_E(C)
\end{bmatrix}
\]
with determinant $q^2d_E(C)-q^2b_i^2$. Since $w_E({\bf a})$ is a sum of real squares, the minimum is always achieved when a coordinate $0$ of a codeword is lifted into $0$, since any other multiple of $q$ would give a higher Euclidean weight. We may thus assume that ${\bm b}$ has $n-w_H({\bm c})$ zero coordinates, so choosing $i$ to be any of them gives $q^2d_E(C)$ as an upper bound. If $q=2$, then we further have that a non-zero coordinate is 1, thus a coordinate $1$ of a codeword is lifted into $1$ to achieve the lowest Euclidean weight. In this case, we may choose $i$ such that $b_i=1$, and tighten the upper bound to
$q^2d_E(C)-q^2$. More generally, among all codewords of $C$ of minimal Euclidean weight, one may pick as ${\bm b}$ the vector which contains the highest $b_i^2$.
We thus have obtained the claimed upper bound, which is an upper bound not only because we considered only one scenario within each of the cases (a) ${\bm a},{\bm b} \in \rho^{-1}({\bf 0})$, (b)  only ${\bm a}$ or ${\bm b}$ is in $\rho^{-1}({\bf 0})$, but we did not yet consider the case ${\bm a},{\bm b} \in \rho^{-1}({\bm c})$
for ${\bm c}$ non-zero.
\end{proof}
We will see in Remark \ref{rem:l2tight} that this bound is actually tight.

The next result is not specific to lattices from codes, but will turn out to be useful later on.

\begin{Lemma}\label{lem:evennorm}
Given a lattice whose Gram matrix has integer even diagonal entries (the lattice is said to be even), the determinant of a sublattice of rank 2 cannot be less than 3.
\end{Lemma}
\begin{proof}
Let ${\bm B}$ be the generator matrix containing the two linearly independent lattice points ${\bm x},{\bm x}'$ as rows, it will have determinant
\[
||{\bm x}||^2 ||{\bm x}'||^2-\langle {\bm x},{\bm x}' \rangle^2
=4a-b^2,~a,b \in \ZZ.
\]
The smallest positive such integer is 3. Indeed, $4a-1=b^2$ and $4a-2=b^2$ are not possible, since modulo $4$, neither $-1$ nor $-2$ are squares.
\end{proof}

Let $C^\perp=\{ \bm a\in \ZZ_{q}^{n}|\langle {\bm a}, {\bm c}\rangle=0, \forall \bm c\in C\}$ be the dual of $C$. Recall that the dual lattice $\Lambda^*_C$ of $\Lambda_C$ is  $\Lambda^*_C=\{\bm v\in \RR^{n}|\langle {\bm v}, {\bm d}\rangle\in \ZZ,\forall \bm d\in \Lambda_C\}$.  The following is known (see Theorem 2, p.183 for $q=2$ \cite{splag}) and recalled for the sake of completeness.

\begin{Proposition}
\label{prop:dual}
With the notations above, we have $\frac{1}{\sqrt{q}}\Lambda_{C^\perp}=(\frac{1}{\sqrt{q}}\Lambda_{C})^{\ast}$, or equivalently $\frac{1}{q}\Lambda_{C^\perp}=\Lambda^*_C$. Furthermore, $\frac{1}{\sqrt{q}}\Lambda_{C}$ is unimodular if and only if $C= C^\perp$.
\end{Proposition}
\begin{proof}
For any $(x_1,x_2,\dots,x_n)\in\frac{1}{\sqrt{q}}\Lambda_{C^\perp}$ and
$(y_1,y_2,\dots,y_n)\in\frac{1}{\sqrt{q}}\Lambda_{C}$, there are $(a_1,a_2,\dots,a_n)\in C^\perp$ and  $(b_1,b_2,\dots,b_n)\in C$ such that $\sqrt{q}x_i\equiv a_i\pmod q$ and  $\sqrt{q}y_i\equiv b_i\pmod q$. Thus $q\sum_{i}x_{i}y_{i}\equiv \sum_{i}a_{i}b_{i}\pmod q\equiv 0\pmod q$ and we have
$\frac{1}{\sqrt{q}}\Lambda_{C^\perp}\subseteq(\frac{1}{\sqrt{q}}\Lambda_{C})^{\ast}$. Consequently, there is a matrix ${\bm A}\in {\mathcal{M}}_n(\ZZ)$ such that ${\bm A} {\bm B}={\bm B'}$ where ${\bm B}$ denotes a generator matrix of the lattice $(\frac{1}{\sqrt{q}}\Lambda_{C})^{\ast}$ and ${\bm B}'$ is a generator matrix for the sublattice $\frac{1}{\sqrt{q}}\Lambda_{C^\perp}$. To prove the reverse inclusion, it is thus enough to show that $\det({\bm A})=1$, or equivalently that $\det(\frac{1}{\sqrt{q}}\Lambda_{C^\perp})=\det((\frac{1}{\sqrt{q}}\Lambda_{C})^{\ast})$.

Now by Proposition \ref{prop:detlambdaC}, $\det(\frac{1}{\sqrt{q}}\Lambda_{C})=\frac{q^n}{|C|^2}$ so $\det((\frac{1}{\sqrt{q}}\Lambda_{C})^{\ast})=\frac{|C|^2}{q^n}$. We also have $\det(\frac{1}{\sqrt{q}}\Lambda_{C^\perp})=\frac{q^n}{|C^\perp|^2}$ and $|C||C^\perp|=q^n$, so $\det(\frac{1}{\sqrt{q}}\Lambda_{C^\perp})=\frac{q^n}{|C^\perp|^2}=\frac{q^n}{(q^{n}/|C|)^2}=\frac{|C|^2}{q^n}=\det((\frac{1}{\sqrt{q}}\Lambda_{C})^{\ast})$ which concludes the proof, and we have $\frac{1}{\sqrt{q}}\Lambda_{C^\perp}=(\frac{1}{\sqrt{q}}\Lambda_{C})^{\ast}$.

The equivalent statement follows from remembering that $(\frac{1}{\sqrt{q}}\Lambda_{C})^{\ast} = \sqrt{q}\Lambda_{C}^{\ast}$.
\end{proof}

Recall that  $\gamma'_{n,l}(\Lambda)=\sqrt{d_l(\Lambda)d_l(\Lambda^*)}$ and
    $\gamma'_{n,l} = \sup_{rank(\Lambda)=n}\gamma'_{n,l}(\Lambda)$.

 \begin{Proposition}\label{duiou}
 We have $d_{l}(\Lambda_{C}^{\ast})=\frac{1}{q^{2l}}d_{l}(\Lambda_{C^{\perp}})$ so
 \[
 \gamma'_{n,l}(\Lambda_{C})=\frac{1}{q^{l}}\sqrt{d_l(\Lambda_{C})d_l(\Lambda_{C^{\perp}})}.
 \]
In particular, if $C$ is self-dual then $\gamma'_{n,l}(\Lambda_{C})=\frac{d_l(\Lambda_{C})}{q^{l}}$.
 \end{Proposition}
\begin{proof}
This follows from Proposition \ref{prop:dual}.
\end{proof}

\begin{Corollary}
We have $d_{l}(\Lambda_{C}^{\ast})\leq 1$ and for $l=1$, we have $d_{1}(\Lambda_{C}^{\ast})=\min (1,d_{E}(C^{\perp})/q^{2})$. Especially, if $C$ is self dual, then $\gamma'_{n,1}(\Lambda_{C})=\min(q,\frac{d_{E}(C)}{q})$.
\end{Corollary}
\begin{proof}
The first claim is immediate from Corollary \ref{cor:ub}, while the claim for $l=1$ follows from Corollary \ref{cor:l1}. If $C=C^\perp$, then
\[
\gamma'_{n,1}(\Lambda_{C})=\frac{d_1(\Lambda_{C})}{q} = \frac{1}{q}\min (q^2,d_{E}(C))
\]
as desired.
\end{proof}

%
%
%

\section{Results on $\gamma_{n,l}(\Lambda_C)$}
\label{sec:gamma}

In this section, we consider specific codes $C$, their corresponding lattice $\Lambda_C$, and the resulting values of $\gamma_{n,l}(\Lambda_C)$.

\begin{Proposition}
\label{prop:Dnq}
Consider the linear code $C=\{(c_1,\ldots,c_{n-1},\sum_{i=1}^{n-1}c_i),~c_1,\ldots,c_{n-1}\in \ZZ_q\}$. Then $\Lambda_C$ satisfies the following:
\begin{enumerate}
\item $\det(\Lambda_C)=q^2$,
\item $d_1(\Lambda_C)=d_E(C)=2$,
\item $\gamma_{n,1}(\Lambda_C)=\frac{2}{q^{2/n}}  $.
\end{enumerate}
\end{Proposition}
\begin{proof}
\begin{enumerate}
\item This is immediate from Proposition \ref{prop:detlambdaC}, since $q^n/|C|=q$.
\item We have that $d_H(C)=2$, achieved, among others, by ${\bm c}=(c,-c,0,\ldots,0)$
and  ${\bm c'}=(c,0,\ldots,0,c)$, for $c\neq 0$. If $q$ is odd, either $c$ or $q-c$ belongs to $\{1,\ldots,(q-1)/2\}$. If $q$ is even, either $c$ or $q-c$ belongs to $\{1,\ldots,(q-1)/2\}$, or $c=q-c=q/2$. Thus for $q$ odd and $q$ even respectively, $2 \leq w_E({\bm c}) = w_E({\bm c}') \leq 2(\tfrac{q-1}{2})^2 < q^2$,
$2 \leq w_E({\bm c}) = w_E({\bm c}') \leq 2(\tfrac{q}{2})^2 < q^2$, where the lower bound is achieved with $c=1$, and from Corollary \ref{cor:l1}
\[
d_1(\Lambda_C) = \min(q^{2},d_E(C)) = d_E(C) = 2.
\]
\item
For $l=1$, we compute using 1. and 2. that
\[
\gamma_{n,1}(\Lambda_C) = \frac{d_1(\Lambda_C)}{\det(\Lambda_C)^{1/n}}
= \frac{2}{q^{2/n}}.
\]
\end{enumerate}
\end{proof}

The case $q=2$ is particular (the code $C$ is then the single parity check code).

\begin{Proposition}\label{prop:Dn}
Consider the linear code $C=\{(c_1,\ldots,c_{n-1},\sum_{i=1}^{n-1}c_i),~c_1,\ldots,c_{n-1}\in \ZZ_q\}$. Then $\Lambda_C$ satisfies the following:
\begin{enumerate}
\item $\Lambda_C=D_n$ when $q=2$.
\item for $l=1$,
$\gamma_{n,1}(D_n)  \geq \gamma_{n,1}(\Lambda_C)$, meaning that the case $q=2$ provides the best $\gamma_{n,1}(\Lambda_C)$ among choices of $q$.
\item $\gamma_{n,1}(D_n)$ is optimal for $n=3,4,5$,
\item
for $l=2$ and $n\geq 3$, $d_2(\Lambda_C)=3$,
\[
\gamma_{n,2}(D_n) =  \frac{3}{4^{2/n}}\geq \frac{}{} \gamma_{n,2}(\Lambda_C)=\frac{3}{q^{2/n}}
\]
and $\gamma_{n,2}(D_n)$ is optimal for $n=4$  while it serves as a benchmark for other values of $n$, namely
\[
\gamma_{n,2} \geq \frac{3}{4^{2/n}},~n \geq 3.
\]
\end{enumerate}
\end{Proposition}

\begin{proof}
\begin{enumerate}
\item
The lattice $D_n$, by definition, contains integer vectors whose component sum is even. Lattice points in $\Lambda_C$ for $q=2$ are for the form $(c_1+2m_1,\ldots,c_{n-1}+2m_{n-1}, \sum_{i=1}^{n-1} c_i + 2m_n)$ for ${\bm c}=(c_1,\ldots,c_n)\in C$, thus
\[
\sum_{i=1}^{n-1} (c_i+2m_i) + \sum_{i=1}^{n-1} c_i + 2m_n =
2\sum_{i=1}^{n-1} c_i + 2\sum_{i=1}^{n} m_i
\]
and  $\Lambda_C \subseteq D_n$ (this is not true for $q$ odd) and conversely every ${\bm x}=(x_1,\ldots,x_n)$ such that $\sum_{i=1}^n x_i$ is even can be parameterized by letting $x_1,\ldots,x_{n-1}$ free and asking that $x_n=\sum_{i=1}^{n-1} x_i+m$ with $m$ even, so $\rho({\bm x}) \in C$.
\item
For a given $n$, the largest $\gamma_{n,1}(\Lambda_C)=\frac{2}{q^{2/n}}$ is obtained for the smallest $q$, namely $q=2$.
\item
The optimal values for $\gamma_{n,1}=\gamma_n$ are
\[
2^{1/3} = \frac{2}{2^{2/3}},~2^{1/2}= \frac{2}{2^{2/4}},~8^{1/5}=\frac{2}{2^{2/5}}
\]
for $n=3,4,5$.
\item
For $q=2$ and $l=2$, take ${\bm c}=(1,1,0,\ldots,0)$
and  ${\bm c'}=(1,0,\ldots,0,1)$ in $C$. Choose the sublattice generated by
\[
{\bm B}' =
\begin{bmatrix}
1 & 1 & 0 &\ldots&0 \\
1 & 0 & \ldots & 0 &1 \\
\end{bmatrix}
\]
in $\Lambda_C$. Then
\[
\det({\bm B}'{(\bm B}')^T) = 3.
\]
Now since a generator matrix for $C$ is $[I_{n-1}| {\bf 1}]$, a generator and Gram matrices of $\Lambda_C$ are respectively
\[
{\bm B}=
\begin{bmatrix}
I_{n-1} & {\bf 1}_{n-1,1} \\
  {\bf 0}_{1,n-1}      & 2
\end{bmatrix},~
{\bm G}
={\bm B}{\bm B}^T
=
\begin{bmatrix}
I_{n-1}+{\bf 1}_{n-1,n-1} & {\bf 1}_{n-1,1} \\
  {\bf 1}_{1,n-1}      & 4
\end{bmatrix}.
\]
Two arbitrary lattice points
 have even norm, thus invoking Lemma \ref{lem:evennorm} gives
\[
\gamma_{n,2}(\Lambda_C) = \frac{d_2(\Lambda_C)}{\det(\Lambda_C)^{2/n}}
= \frac{3}{4^{2/n}}.
\]
When $n=4$, we get $3/2$, the optimal value. Finally, the sublattice generated by
\[
{\bm B}' =
\begin{bmatrix}
1 & -1 & 0 &\ldots&0 \\
1 & 0 & \ldots & 0 &1 \\
\end{bmatrix}
\]
is in $\Lambda_C$ for $q>2$, so $d_2(\Lambda_C)\leq 3$.

Take two lattice vectors ${\bm a},{\bm b}$, which are linearly independent, so they generate a sublattice $\Lambda'$ of $\Lambda_{C}$ of rank 2.
A generator matrix is
\[
{\bm B}' =
\begin{bmatrix}
{\bm a} \\
{\bm b}
\end{bmatrix}
=
\begin{bmatrix}
a_1 & \ldots & a_n \\
b_1 & \ldots & b_n
\end{bmatrix}
\begin{bmatrix}
I_{n-1} & {\bf 1}_{n-1,1} \\
  {\bf 0}_{1,n-1}      & q
\end{bmatrix}
=
\begin{bmatrix}
a_1 & \ldots & a_{n-1} & \alpha \\
b_1 & \ldots & b_{n-1} &  \beta
\end{bmatrix},
\]
for
$\alpha=\sum_{i=1}^{n-1}a_i+ qa_n$,
$\beta=\sum_{i=1}^{n-1}b_i+ qb_n$.
Then
\[
{\bm B}'({\bm B}')^T =
\begin{bmatrix}
\sum_{i=1}^{n-1}a_i^2 + \alpha^2& \sum_{i=1}^{n-1}a_ib_i + \alpha\beta \\
\sum_{i=1}^{n-1}a_ib_i + \alpha\beta & \sum_{i=1}^{n-1}b_i^2+ \beta^2
\end{bmatrix}
\]
with determinant
\begin{eqnarray*}
A_n&=&(\sum_{i=1}^{n-1}a_i^2 + \alpha^2)(\sum_{i=1}^{n-1}b_i^2+ \beta^2)-(\sum_{i=1}^{n-1}a_ib_i + \alpha\beta)^2
\\
&=&
\sum_{i=1}^{n-1}a_i^2\sum_{i=1}^{n-1}b_i^2 + \beta^2\sum_{i=1}^{n-1}a_i^2+\alpha^2\sum_{i=1}^{n-1}b_i^2
-
(\sum_{i=1}^{n-1}a_ib_i )^2-2\alpha\beta\sum_{i=1}^{n-1}a_ib_i \\
& = &
\sum_{i=1}^{n-1}a_i^2\sum_{i=1}^{n-1}b_i^2-(\sum_{i=1}^{n-1}a_ib_i )^2+
\sum_{i=1}^{n-1}( \beta a_i - \alpha b_i)^2
\end{eqnarray*}
where
\begin{eqnarray*}
\sum_{i=1}^{n-1}( \beta a_i - \alpha b_i)^2
&=& \sum_{i=1}^{n-1}\left( ( \sum_{j=1}^{n-1}b_j+ qb_n) a_i - (\sum_{j=1}^{n-1}a_j+ qa_n) b_i\right)^2 \\
&=&
\sum_{i=1}^{n-1}\left(  \sum_{j=1}^{n-1}(b_ja_i-a_jb_i )+ q(b_na_i - b_ia_n) \right)^2 \\
&=&
\sum_{i=1}^{n-1}\left( \sum_{j=1}^{n-1}D_{ij}+qD_{in}\right)^{2},
\end{eqnarray*}
 setting
\[
D_{ij}=
\det
\begin{bmatrix}
a_i  & a_j  \\
b_i  & b_j
\end{bmatrix}.
\]
\end{enumerate}

We are left to prove that
\begin{eqnarray*}
A_n&=&\sum_{i=1}^{n-1}a_i^2\sum_{i=1}^{n-1}b_i^2-(\sum_{i=1}^{n-1}a_ib_i )^2+
\sum_{i=1}^{n-1}\left( \sum_{j=1}^{n-1}D_{ij}+qD_{in}\right)^2
\\
&=&
\sum_{1\leq i<j\leq n-1}D_{ij}^{2}+
\sum_{i=1}^{n-1}\left( \sum_{j=1}^{n-1}D_{ij}+qD_{in}\right)^2 \geq 3.
\end{eqnarray*}

For $A_n$ to be 1, we need exactly one square to 1, and the others to be $0$.
\begin{itemize}
\item If $\sum_{1\leq i<j\leq n-1}D_{ij}^{2}=0$, then $D_{ij}=0$ for all $1\leq i,j\leq n-1$, which implies $A_{n}=q^{2}\sum_{i=1}^{n}D_{in}^2$,  a contradiction.
\item If $\sum_{1\leq i<j\leq n-1}D_{ij}^{2}=1$, then $D_{i_{0}j_{0}}=\pm 1$ for some $1\leq i_{0},j_{0}\leq n-1$, and $D_{ij}=0$ for all $1\leq i,j\leq n-1$ and $\{i,j\}\neq\{i_{0},j_{0}\}$. Then $(D_{i_{0}j_{0}}+qD_{i_{0}n})^{2}=0$ implies $qD_{i_{0}n}=\pm 1$, a contradiction.
\end{itemize}

For $A_n$ to be 2, we need exactly two squares to 1, and the others to be $0$. For the same reason as above, we know there is at least one $D_{ij}=\pm 1$ for some $1\leq i<j\leq n-1$.
\begin{itemize}
\item If $\sum_{1\leq i<j\leq n-1}D_{ij}^{2}=2$, then there are just
two $D_{ij}=\pm 1$ for $1\leq i<j\leq n-1$ say $D_{i_{0}j_{0}}=D_{i_{1}j_{1}}=\pm 1$. If $i_{0}= i_{1}$ then $D_{j_{0}i_{0}}+qD_{j_{0}n}=\pm 1+qD_{j_{0}n}=0$, a contradiction. If $i_{0}\neq i_{1}$, the case $j_{0}= j_{1}$ is the same as above, thus suppose $j_{0}\neq j_{1}$. Then $D_{i_{0}j_{0}}+qD_{i_{0}n}=\pm 1+qD_{i_{0}n}=0$, a contradiction.
\item If $\sum_{1\leq i<j\leq n-1}D_{ij}^{2}=1$, then $D_{i_{0}j_{0}}=\pm 1$ for some $1\leq i_{0},j_{0}\leq n-1$, and $D_{ij}=0$ for all $1\leq i,j\leq n-1$ and $\{i,j\}\neq\{i_{0},j_{0}\}$. Then either $(D_{i_{0}j_{0}}+qD_{i_{0}n})^{2}=1$ and $(D_{j_{0}i_{0}}+qD_{j_{0}n})^{2}=0$ or  $(D_{i_{0}j_{0}}+qD_{i_{0}n})^{2}=0$ and $(D_{j_{0}i_{0}}+qD_{j_{0}n})^{2}=1$, which implies either $qD_{i_{0}n}=\pm 1$ or $qD_{j_{0}n}=\pm 1$, a contradiction.
\end{itemize}
Thus $A_{n}=3$ for $n\geq 3$.
\end{proof}

\begin{Corollary}\label{cor:n57}
We have
\[
1.723 \approx \frac{3}{4^{2/5}} \leq \gamma_{5,2} \leq  2,~
2.0189 \approx\frac{3}{4^{2/7}} \leq
\gamma_{7,2} \leq  2^{5/3} \approx 3.1748.
\]
\end{Corollary}
\begin{proof}
The lower bound follows from the above proposition, namely
\[
\gamma_{n,2}\geq\frac{3}{4^{2/n}}.
\]
The upper follows from the known inequality
$(\gamma_{n,l})^{n-2l}\leq(\gamma_{n-l,l})^{n-l}$. For $n=5$, we have
\[
\gamma_{5,2}\leq(\gamma_{3,2})^3 =2
\]
while for $n=7$
\[
(\gamma_{7,2})^{3}\leq(\gamma_{5,2})^{5} \leq 2^5.
\]
\end{proof}

Let $\mathcal{R}(r,m)$ denote the binary Reed-Mueller code of order $r$ and length $2^m$. A recursive formula for a generator matrix of $\mathcal{R}(r,m)$ is
\[
{\bm B}_{r,m}=
\begin{bmatrix}
{\bm B}_{r,m-1} & {\bm B}_{r,m-1} \\
{\bf 0} & {\bm B}_{r-1,m-1}
\end{bmatrix},
\]
with ${\bm B}_{0,i} = {\bf 1}_{1,2^i}$ and  ${\bm B}_{i,i} = I_{2^i}$ for $i\leq m$.
It is known \cite{book} that the dimension $k$ of this code is $k=\sum_{i=0}^r {m \choose i }$ and that $\mathcal{R}(i,m) \subseteq \mathcal{R}(j,m)$ for $i\leq j$.
Consider the lattice $\Lambda_{\mathcal{R}(j,m)}$. It will then contain the points $\rho^{-1}({\bm c})$ for ${\bm c}\in \mathcal{R}(i,m)$, as well as the sublattice generated by ${\bm B}_{i,m} \subset \rho^{-1}{(\bm B}_{i,m})$, for ${\bm B}_{i,m}$ a generator matrix of $\mathcal{R}(i,m)$, which we will denote by $\Lambda({\bm B}_{i,m})$. There is of course no reason for this sublattice to give the smallest determinant, but the construction is available for an arbitrarily large family of parameters, which we will explore next.

From (\ref{eq:gammaLLk}), we are interested in computing
\[
\gamma(\Lambda_{\mathcal{R}(j,m)},\Lambda({\bm B}_{i,m})) =
\frac{\det(\Lambda({\bm B}_{i,m}))}{\det(\Lambda_{\mathcal{R}(j,m)})^{l/n}},~l=\dim(\mathcal{R}(i,m)) = \sum_{i=0}^r {m \choose i }.
\]

Numerical values of $\det(\Lambda({\bm B}_{r,m}))$ are computed in Table \ref{table:bmr}.  When $r=m$, a generator matrix is the identity, so $\det(I_{2^m})=1$. The case $r=m$ is thus not listed in Table \ref{table:bmr}.

When $r=0$, a generator matrix is ${\bf 1}_{2^m}$,
so ${\bf 1}{\bf 1}^T=2^m$ and $\det({\bf 1}{\bf 1}^T)=2^m$.

\begin{table}
\centering
\begin{tabular}{ccccc}
$m$ & $r$ & $k=\sum_{i=0}^r {m \choose i }$ & $\det(\Lambda({\bm B}_{r,m}))$ & $\det(\Lambda_{\mathcal{R}(r,m)})= (2^{2^m-k})^2 $ \\
\hline
1 & 0 & 1  & 2 & $(2^{2-k})^2=2^2$  \\
\hline
2 & 0 & 1 & 4 & $(2^{4-k})^2=2^6$  \\
  & 1 & 3 & 4 & $2^2$\\
\hline
3 & 0 & 1 &  $8 $ &$(2^{8-k})^2=2^{14}$  \\
  & 1 & 4  &  64 & $2^8$\\
  & 2 & 7 &  8  & $2^{2}$\\
\hline
4 & 0 & 1 & $16$ & $(2^{16-k})^2=2^{30}$   \\
  & 1 & 5 & 4096 & $2^{22}$\\
  & 2 & 11 &  4096 & $2^{10}$\\
  & 3 & 15 &  16 & $2^{2}$\\
\hline
5 & 0 & 1 &$32$ & $(2^{32-k})^2=2^{62}$  \\
  & 1 & 6 &1048576&  $2^{52}$ \\
  & 2 & 16 &1073741824& $2^{32}$  \\
  & 3 & 26 &1048576&  $2^{12}$ \\
  & 4 & 31 &32& $2^{2}$  \\
\end{tabular}
\caption{\label{table:bmr}
Parameters associated with the sublattice $\Lambda({\bm B}_{r,m})$. By Proposition \ref{prop:detlambdaC},
$\det(\Lambda_C)=\left(\frac{q^n}{|C|}\right)^2 $.
}
\end{table}

\begin{Proposition}\label{prop:rmr1}
If $r=1$, then
\[
\det({\bm B}_{1,m}{\bm B}_{1,m}^T)
=
4\cdot 2^{(m-2)(m+1)},~m\geq 2.
\]
Furthermore, if ${\bm B}'$ is a submatrix of ${\bm B}_{1,m}$ containing $l$ of its rows, then
\[
\det({\bm B}'({\bm B}')^T)
=
\left\{
\begin{array}{ll}
4\cdot 2^{(m-2)l} & \mbox{ if row 1 is used} \\
(1+l)  \cdot 2^{(m-2)l}   & \mbox{ else}
\end{array}
\right.
\]
All matrices ${\bm B}_{1,m}{\bm B}_{1,m}^T$ and ${\bm B}'({\bm B}')^T$ have integer even diagonal entries.
\end{Proposition}

\begin{proof}
For $r=1$, $k=\sum_{i=0}^1 {m \choose i }=1+m$.

We have for $m=2$
\[
{\bm B}_{1,2}
=
\begin{bmatrix}
 I_2 &   I_2 \\
 {\bf 0}_{1,2} & {\bf 1}_{1,2}
\end{bmatrix} \in \mathcal{M}_{3,4}(\mathbb{R}),~
{\bm B}_{1,2}{\bm B}_{1,2}^T=
\begin{bmatrix}
 2 I_2 &   {\bf 1}_{2,1} \\
 {\bf 1}_{1,2} & 2
\end{bmatrix},
\]
for $m=3$
\[
{\bm B}_{1,3}
=
\begin{bmatrix}
 I_2 &   I_2  &   I_2 &   I_2 \\
 {\bf 0}_{1,2} & {\bf 1}_{1,2} &  {\bf 0}_{1,2} & {\bf 1}_{1,2} \\
 {\bf 0}_{1,2} & {\bf 0}_{1,2} &  {\bf 1}_{1,2} & {\bf 1}_{1,2} \\
\end{bmatrix}\in \mathcal{M}_{3,4}(\mathbb{R}),~
{\bm B}_{1,3}{\bm B}_{1,3}^T
=
\begin{bmatrix}
4I_2 &   {\bf 2}_{2,1} & {\bf 2}_{2,1} \\
{\bf 2}_{1,2} & 4 &  2 \\
{\bf 2}_{1,2} & 2 & 4
\end{bmatrix}
\]
and more generally
\begin{eqnarray*}
  {\bm B}_{1,m}&=&
\begin{bmatrix}
{\bm B}_{1,m-1} & {\bm B}_{1,m-1} \\
{\bf 0} & {\bm B}_{0,m-1}
\end{bmatrix}\\
&=&
\begin{bmatrix}
{\bm B}_{1,m-2} & {\bm B}_{1,m-2} & {\bm B}_{1,m-2} & {\bm B}_{1,m-2}  \\
{\bf 0} &  {\bm B}_{0,m-2} & {\bf 0} &  {\bm B}_{0,m-2} \\
{\bf 0}  & & {\bf 1}_{2^{m-1}} &
\end{bmatrix}\\
& = &
\begin{bmatrix}
{\bm B}_{1,2} & {\bm B}_{1,2} & \ldots & {\bm B}_{1,2} & {\bm B}_{1,2}  \\
\vdots &&&& \vdots\\
{\bf 0} &  {\bf 1}_{2^{m-2}} & {\bf 0} &  {\bf 1}_{2^{m-2}} \\
{\bf 0}  & & {\bf 1}_{2^{m-1}} &
\end{bmatrix}
\end{eqnarray*}
contains $2^{m-1}$ block $I_2$ occupying the first two rows, followed by $m-1$ rows, with row $m+1-i$ repeating the pattern ${\bf 0}, {\bf 1}_{2^{m-1-i}}$ $2^i$ times, $i=0,\ldots,m-2$.
We may further add the second row to the first one, to obtain a first row comprising only 1 as a coefficient, which we denote by ${\bm B}'_{1,m}$.
Thus
\[
{\bm B}'_{1,m}({\bm B}')_{1,m}^T
=
\begin{bmatrix}
2^{m} &  2^{m-1} &  \ldots      & 2^{m-1} \\
2^{m-1}     & 2^{m-1}  &        &  \\
2^{m-1} & 2^{m-2} & \ldots & 2^{m-2} \\
 \vdots     &   2^{m-2}       & \ddots &  2^{m-2} \\
2^{m-1}     &   \ldots  & 2^{m-2}       & 2^{m-1} \\
\end{bmatrix}
\]
that is all coefficients are $2^{m-2}$, apart from the diagonal, the first row and the first column, where all the coefficients are $2^{m-1}$ and the coefficient in the first row first column is $2^m$. We then have that
\[
\det({\bm B}_{1,m}{\bm B}_{1,m}^T)
=
\det({\bm B}'_{1,m}({\bm B}')_{1,m}^T
)
=
(2^{m-2})^{m+1}
\det
\begin{bmatrix}
4 &  2 &  \ldots      & 2 \\
2     & 2  &        &  \\
 \vdots     &          & \ddots &  1 \\
2     &   \ldots  & 1       & 2 \\
\end{bmatrix}.
\]
We are left to consider this determinant. Subtract twice the last row to the first row, and compute the minor corresponding to the last entry of the first row (the only non-zero coefficient left) to get
\[
2
\det
\begin{bmatrix}
2      & 2  &  1 & \ldots  & 1       \\
\vdots & 1  &      \ddots &  & \\
2      & 1  &  \ldots    && 2      \\
2 & 1 & \ldots && 1
\end{bmatrix}
\]
and subtract the last row to all the rows above. Coefficients of the first column are canceled, those that are 1 are canceled. This gives
\[
2
\det
\begin{bmatrix}
0     & 1  &  0 & \ldots  & 0       \\
\vdots & 0  &      \ddots &  & \\
0      & 0  &  \ldots    && 1      \\
2 & 1 & \ldots && 1
\end{bmatrix}
\]
and computing the minor corresponding to the first entry of the last row completes the computation by showing that this determinant is 4, as desired.

For the case of ${\bm B}'$, we need to compute instead
\[
\det({\bf 1}_{l,l}+I_l) = (2+(l-1))
\]
for the case where row 1 is not used, by adding all rows to the first row, extracting $2+(l-1)$ by multilinearity and finally  subtracting the first row from each of the remaining rows.
\end{proof}
\begin{Corollary}\label{cor:l2m3}
With the notations of the proposition,
\[
\min\{(1+l)  \cdot 2^{(m-2)l},  4\cdot 2^{(m-2)l}\} =
\left\{
\begin{array}{ll}
3 \cdot 2^{2(m-2)} & l=2\\
4\cdot 2^{(m-2)l} & l \geq 3
\end{array}
\right.
\]
and for $l=2$, $3 \cdot 2^{2(m-2)}\leq 2^{2l}$ whenever $m\leq 3$, while for $l\geq 3$, $4\cdot 2^{(m-2)l} \leq 2^{2l}$ whenever $m \leq  \frac{2l-2}{l}+2$.
\end{Corollary}
\begin{proof}
We have
\[
3 \cdot 2^{2(m-2)} < 2^4 \iff 3 < 2^{4-2(m-2)}
\iff 4-2(m-2) \geq 2 \iff m \leq 3
\]
and
\[
4\cdot 2^{(m-2)l} \leq 2^{2l} \iff
2+(m-2)l \leq 2l \iff m \leq  \frac{2l-2}{l}+2.
\]
\end{proof}

\begin{Proposition}\label{prop:rmm-1}
Suppose $r=1$.
\begin{enumerate}
\item For $l=1$,
\[
\gamma_{n,1}(\Lambda_{\mathcal{R}(1,m)}) =  \left\{
\begin{array}{ll}
\frac{2^{3/2}}{2} = \sqrt{2} & m=2 \\
2^{2(m+1)/2^m} & m \geq 3,
\end{array}
\right.
\]
which is optimal for $m=2,3$ (that is in dimension 4 and 8).
\item
For $l=2$, there exists a sublattice $\Lambda({\bm B})$ such that
\[
\gamma(\Lambda_{\mathcal{R}(1,m)},\Lambda({\bm B}))=
\left\{
\begin{array}{ll}
3 & m=3~(n=8) \\
\frac{1}{2\sqrt{2}}  & m = 4~(n=16) \\
\end{array}
\right.
\]
Furthermore,
$\gamma_{8,2}(\Lambda_{\mathcal{R}(1,3)})=3$, which is optimal in dimension 8 for $l=2$.
\item
For $l\geq 3$, there exists a sublattice $\Lambda({\bm B})$ such that
\[
\gamma(\Lambda_{\mathcal{R}(1,m)},\Lambda({\bm B}))=
\frac{4\cdot 2^{(m-2)l}}{\left( \frac{2^{2^m}}{2^{m+1}} \right)^{2l/n}}
\]
for $m\leq \tfrac{2l-2}{l}+2$.
This quantity simplifies to 4 when $m=3$ and $l=3,4$.
\end{enumerate}
\end{Proposition}
\begin{proof}
The code $\mathcal{R}(1,m)$ has length $n=2^m$, so the ambient space has dimension $n=2^m$, and has dimension $\sum_{i=0}^{1} {m \choose i }=m+1$, so $|\mathcal{R}(1,m)|=2^{m+1}$ and by Proposition \ref{prop:detlambdaC},
\[
\det(\Lambda_{\mathcal{R}(1,m)})= \left(\frac{2^{2^m}}{2^{m+1}} \right)^2.
\]
Furthermore for $m=3$, $\mathcal{R}(r,m)^\perp=\mathcal{R}(m-r-1,m)$ so $\mathcal{R}(1,3)^\perp=\mathcal{R}(1,3)$. By Proposition \ref{prop:dual}, $\frac{1}{\sqrt{2}}\Lambda_{\mathcal{R}(1,3)}$ is unimodular, which implies that even after dividing by 4, a Gram matrix still has integer coefficients.
\begin{enumerate}
\item
For $l=1$, we know from Corollary \ref{cor:l1} that
\[
\gamma_{n,1}(\Lambda_{\mathcal{R}(1,m)}) =  \min(q^{2},d_E(C))\frac{|C|^{2/n}}{q^{2}}
=
\min(q^{2},d_E(C))
\frac{2^{2(m+1)/2^m}}{2^2}.
\]
The minimum distance is $2^{m-1}$, so for $m=2$, there is a codeword of weight 2 which will yield a lattice point of Euclidean weight 2, so
\[
\gamma_{n,1}(\Lambda_{\mathcal{R}(1,m)}) =
\left\{
\begin{array}{ll}
\frac{2^{3/2}}{2} = \sqrt{2} & m=2 \\
2^{2(m+1)/2^m} & m \geq 3.
\end{array}
\right.
\]
\item
Using Corollary \ref{cor:l2m3}, for $l=2$ and $m=3$, there is a sublattice $\Lambda({\bm B})$ of rank $l=2$, which is even (see Proposition \ref{prop:rmr1}), such that
\[
\gamma(\Lambda_{\mathcal{R}(1,3)},\Lambda({\bm B}))=
\frac{3\cdot 2^2}{
\det(\Lambda_{\mathcal{R}(1,m)})^{l/n}}
= \frac{3 \cdot 2^2}{\left(2^4 \right)^{1/2}} = 3.
\]
Since $\frac{1}{\sqrt{2}}\Lambda_{\mathcal{R}(1,3)}$ is unimodular, $\det(\frac{1}{\sqrt{2}}\Lambda_{\mathcal{R}(1,3)})=1$, and we have that even after normalizing by $1/4$, the sublattice is still even (codewords have even weights). It then follows from Lemma \ref{lem:evennorm} that
\[
\gamma_{8,2}(\tfrac{1}{\sqrt{2}}\Lambda_{\mathcal{R}(1,3)}) = 3.
\]
That $\gamma_{8,2}(\Lambda_{\mathcal{R}(1,3)}) = 3$ follows from the fact that the constant is scaling invariant.

For $l=2$ and $m=4$,  $3 \cdot 2^{4} > 2^{4}$, $4\cdot 2^4 > 2^{4}$, so a smaller bound comes from the sublattice $2\ZZ^{2}$, so
\[
\gamma(\Lambda_{\mathcal{R}(1,4)},2\ZZ^2) =
\frac{ 2^4}{
\det(\Lambda_{\mathcal{R}(1,m)})^{2/n}}
= \frac{ 2^4}{\left( 2^{11} \right)^{1/2}}
= 2^{-3/2}.
\]

\item
For $l\geq 3$, the same corollary gives a  sublattice $\Lambda({\bm B})$ of rank $l$ such that
\[
\gamma(\Lambda_{\mathcal{R}(1,m)},\Lambda({\bm B})) =
\frac{4\cdot 2^{(m-2)l}}{
\det(\Lambda_{\mathcal{R}(1,m)})^{l/n}}
= \frac{4\cdot 2^{(m-2)l}}{\left( \frac{2^{2^m}}{2^{m+1}} \right)^{2l/n}}
\]
for $m\leq \tfrac{2l-2}{l}+2$.
For $l=m=3$, this simplifies to
\[
\gamma(\Lambda_{\mathcal{R}(1,m)},\Lambda({\bm B})) =
\frac{4\cdot 2^{3}}{
\det(\Lambda_{\mathcal{R}(1,m)})^{l/8}}
= \frac{4\cdot 2^{3}}{\left(2^{4}\right)^{6/8}}
= 4.
\]
Similarly for $l=4$ and $m=3$
\[
\gamma(\Lambda_{\mathcal{R}(1,m)},\Lambda({\bm B})) =
\frac{4\cdot 2^{4}}{
\det(\Lambda_{\mathcal{R}(1,m)})^{l/n}}
= \frac{4\cdot 2^4}{2^4} = 4.
\]
\end{enumerate}
\end{proof}

\begin{Remark}\label{rem:l1tight}
That $\gamma_{8,1}(\Lambda_{\mathcal{R}(1,3)})=2$ shows that the bound of Corollary \ref{cor:ub}, namely $\gamma_{8,1}(\Lambda_C)\leq|C|^{2l/8}$, is tight for $l=1$, since $|C|^{1/4}=2$.
\end{Remark}

\begin{Remark} \label{rem:l2tight}
For $l=2$, Proposition \ref{prop:ubl2} gives
\[
\gamma_{8,2}(\Lambda_{\mathcal{R}(1,3)}) \leq \min (1, \tfrac{1}{q^2}(d_E(C)-b^2))|C|^{4/n}
= 4\min (1,\tfrac{3}{4}) = 3
\]
where $b=1$, since the code is binary and $|b|$ is maximal among $b_1,\ldots,b_n$ means maximal between $0$ and $1$.
That $\gamma_{8,2}(\Lambda_{\mathcal{R}(1,3})=3$ shows that the bound of Proposition \ref{prop:ubl2}  is tight for $l=2$.
\end{Remark}

The case $r=1,m=3$, that is $\mathcal{R}(1,3)$, is special.
The binary Hamming code $\mathcal{H}_3$ of length 7 can be extended into $\tilde{\mathcal{H}}_3$ of length 8, which is equivalent to $\mathcal{R}(1,3)$.
Then
\[
\frac{1}{\sqrt{2}}\rho^{-1}(\tilde{\mathcal{H}}_3)
\]
gives the lattice $E_8$ through Construction A. Indeed, a generator matrix for $\tilde{\mathcal{H}}_3$ is
\[
{\bm Q}=
\begin{bmatrix}
1 & 0 & 0 & 0 & 1 & 1 & 0 & 1\\
0 & 1 & 0 & 0 & 1 & 0 & 1 & 1\\
0 & 0 & 1 & 0 & 0 & 1 & 1 & 1 \\
0 & 0& 0 & 1 & 1 & 1 & 1  & 0 \\
\end{bmatrix}.
\]
Then a generator matrix for $\frac{1}{\sqrt{2}}\rho^{-1}(\tilde{\mathcal{H}}_3)$ is
\[
{\bm P} =
\frac{1}{\sqrt{2}}
\begin{bmatrix}
1 & 0 & 0 & 0 & 1 & 1 & 0 & 1\\
0 & 1 & 0 & 0 & 1 & 0 & 1 & 1\\
0 & 0 & 1 & 0 & 0 & 1 & 1 & 1 \\
0 & 0 & 0 & 1 & 1 & 1 & 1  & 0 \\
0 & 0 & 0 & 0 & 2 & 0 & 0 & 0 \\
0 & 0 & 0 & 0 & 0 & 2 & 0 & 0 \\
0 & 0 & 0 & 0 & 0 & 0 & 2 & 0 \\
0 & 0 & 0 & 0 & 0 & 0 & 0 & 2
\end{bmatrix}
\]
with corresponding Gram matrix
\[
\begin{bmatrix}
2 &  1  & 1 &  1 &  1 &  1 &  0 &  1\\
1 &  2  & 1 &  1 &  1 &  0 &  1 &  1\\
1 &  1  & 2 &  1 &  0 &  1  & 1 &  1\\
1 &  1  & 1 &  2 &  1 &  1 &  1 &  0\\
1 &  1  & 0 &  1 &  2 &  0 &  0 &  0\\
1 &  0  & 1 &  1 &  0 &  2 &  0 &  0\\
0 &  1  & 1 &  1 &  0 &  0 &  2 &  0\\
1 &  1  & 1 &  0 &  0 &  0 &  0 &  2\\
\end{bmatrix}.
\]
It is an even integral matrix with determinant 1, so it is indeed a Gram
matrix for $E_8$.

\begin{Proposition}\label{prop:rmm-1}
When $r=m-1$,
\[
\gamma_{n,1}(\Lambda_{\mathcal{R}(m-1,m)}) =   \frac{2}{2^{2/2^m}}
\]
and
\[
\gamma_{n,l}(\Lambda_{\mathcal{R}(m-1,m)})
\leq
\left\{
\begin{array}{ll}
3\cdot 2^{2(m-2)-4/2^m} & l=2 \\
4\cdot 2^{2(m-2)-2l/2^m} & l\geq 3
\end{array}
\right.
\]
\end{Proposition}
\begin{proof}
The code $\mathcal{R}(m-1,m)$ has length $n=2^m$, so the ambient space has dimension $n=2^m$, and has dimension $\sum_{i=0}^{m-1} {m \choose i }=2^m-1$, so $|\mathcal{R}(m-1,m)|=2^{2^m-1}$ and by Proposition \ref{prop:detlambdaC},
\[
\det(\Lambda_{\mathcal{R}(m-1,m)})= 2^2.
\]
For $l=1$, we know from Corollary \ref{cor:l1} that
\[
\gamma_{n,1}(\Lambda_{\mathcal{R}(m-1,m)}) =  \min(q^{2},d_E(C))\frac{|C|^{2/n}}{q^{2}}
\]
where
\[
\frac{|C|^{2/n}}{q^{2}}=\frac{2^{2-2/n}}{2^2}
=\frac{1}{2^{2/n}}.
\]
The minimum distance is $2^{m-r}=2$, and since a generator matrix is
\[
{\bm B}_{m-1,m}=
\begin{bmatrix}
I_{2^{m-1}} & I_{2^{m-1}}  \\
{\bf 0} & {\bm B}_{r-1,m-1}
\end{bmatrix},
\]
the codeword $(1,0,\ldots,0,1,0,\ldots, 0)$ has Euclidean weight 2, yielding:
\[
\gamma_{n,1}(\Lambda_{\mathcal{R}(m-1,m)}) = \frac{2}{2^{2/n}}.
\]

For $l\geq 2$, we know that $\mathcal{R}(i,m)\subseteq \mathcal{R}(r-1,m)$ for $i\leq m-1$,
and by Proposition \ref{prop:rmr1} that $\mathcal{R}(1,m)$ contains $l$-dimensional subcodes, $l\leq m+1$, whose generator matrix ${\bm B}'$ satisfies
\[
\det({\bm B}'({\bm B}')^T)
=
\left\{
\begin{array}{ll}
4\cdot 2^{(m-2)l} & \mbox{ if row 1 is used} \\
(1+l)  \cdot 2^{(m-2)l}   & \mbox{ else}
\end{array}
\right.
\]
Therefore
\[
\gamma(\Lambda_{\mathcal{R}(m-1,m)},\Lambda({\bm B}')) =
\left\{
\begin{array}{ll}
3\cdot 2^{2(m-2)-4/2^m} & l=2 \\
4\cdot 2^{2(m-2)-2l/2^m} & l\geq 3
\end{array}
\right.
\]
(the case $l=2$ corresponds to not using row 1 in the construction of ${\bm B}'$, while row 1 is used for $l\geq 3$).


\end{proof}

We recognize in $\mathcal{R}(m-1,m)$ the single parity check code, with the corresponding lattice $D_n$ discussed in Proposition \ref{prop:Dn}, where results were available for $\gamma_{n,l}(D_n)$, $l=1,2$ and $n\geq 3$. Here values of $n$ are restricted to powers of 2, so we may compare that both propositions are consistent. From the earlier proposition,
\[
\gamma_{4,1}(D_4) = \sqrt{2},~
\gamma_{4,2}(D_4) = \frac{3}{2}
\]
and the first value matches ($m=2$) while the second is only an upper bound in Proposition \ref{prop:rmm-1}, which we know is an equality thanks to the earlier proposition.


\section{Results on $\gamma'_{n,l}(\Lambda_C)$}\label{sec:gamma'}

We next consider specific codes $C$, their corresponding lattice $\Lambda_C$, and this time the resulting values of $\gamma'_{n,l}(\Lambda_C)$.

Similar to Propositions \ref{prop:Dnq} and \ref{prop:Dn}, we have:

\begin{Proposition}
Consider the linear code $C=\{(c_1,\ldots,c_{n-1},\sum_{i=1}^{n-1}c_i),~c_1,\ldots,c_{n-1}\in \ZZ_q\}$.  Then:
\begin{enumerate}
\item
$d_{1}(\Lambda_{C^{\perp}})=\min(n,q^{2})$ and then $\gamma'_{n,1}(\Lambda_{C})=\frac{1}{q}\sqrt{2\min(n,q^{2})}$. In particular, if $n\geq q^{2}$, then $\gamma'_{n,1}(\Lambda_{C})=\sqrt{2}$ so we have $\gamma'_{n,1}\geq\sqrt{2}$ when $n\geq 4$.
\item For $q=2$, $l=1$, we have $\gamma'_{n,1}(\Lambda_{C})=\frac{1}{2}\sqrt{2\min(n,4)}$. In particular,
\[
\gamma'_{2,1}(\Lambda_{C})=1,~
\gamma'_{3,1}(\Lambda_{C})=\sqrt{3/2},~
\gamma'_{4,1}(\Lambda_{C})=
\gamma'_{5,1}(\Lambda_{C})=\sqrt{2}
\]
and
$\gamma'_{n,1}(\Lambda_{C})$ is optimal when $n=3,4,5$.

\item  For $l=2$, $n\geq 3$, we have $d_{2}(\Lambda_{C^\perp})=\min(q^4,q^2(n-1))$.
\item
For $n\geq 3$, we have
$ \gamma_{n,2}(\Lambda_C)=
\frac{1}{q}\sqrt{3 \min(q^2,n-1)}$, so for $q=2$ and $n=4$,  $\gamma'_{4,2}(\Lambda_{C})=3/2$ is optimal. Furthermore, $\gamma'_{n,2}\geq \sqrt{3}$ for $n\geq 5$.
\end{enumerate}
\end{Proposition}

\begin{proof}
\begin{enumerate}
\item Since $(1,1,\ldots,-1)$ generates $C^{\perp}$, we know that $d_{E}(C^{\perp})=n$. Then by Proposition \ref{duiou}
\[
 \gamma'_{n,1}(\Lambda_{C})=\frac{1}{q}\sqrt{d_1(\Lambda_{C})d_1(\Lambda_{C^{\perp}})}
 =
\frac{1}{q}\sqrt{2\min(q^2,n)}
\]
using Proposition \ref{prop:Dnq} ($d_1(\Lambda_C)=2$) and Corollary \ref{cor:l1} ($d_1(\Lambda_{C^\perp})=\min(q^2,d_E(C^\perp))$). Suppose now $n\geq q^2$, then
\[
\gamma'_{n,1}(\Lambda_{C})=\frac{1}{q}\sqrt{2q^2} = \sqrt{2}.
\]
When $q=2$ and $n\geq 4$, the lattice $\Lambda_C$ provides a lower bound.
\item From 1., if $q=2$ and $n\geq 4$, $\gamma'_{n,1}=\sqrt{2}$. Again from 1., $\gamma'_{n,1}(\Lambda_{C})=\frac{1}{2}\sqrt{2\min(n,4)}$ yields
\[
\gamma'_{2,1}(\Lambda_{C})=1,~
\gamma'_{3,1}(\Lambda_{C})=\frac{1}{2}\sqrt{6}.
\]
\item
Since a generator matrix for $C^\perp$ is $[{\bf 1}~-1]$, a generator and Gram matrices of $\Lambda_{C^\perp}$ are respectively
\[
{\bm B}=
\begin{bmatrix}
1 & {\bf 1}_{1,n-2}& -1\\
{\bf 0}_{n-1,1}  & qI_{n-1} &
\end{bmatrix},~
{\bm G}
={\bm B}{\bm B}^T
=
\begin{bmatrix}
n & q{\bf 1}_{1,n-1} & -q \\
q{\bf 1}_{n-1,1}  & q^2I_{n-1}  & \\
 -q &  &  \\
\end{bmatrix}.
\]

Take two lattice vectors ${\bm a},{\bm b}$, which are linearly independent, so they generate a sublattice $\Lambda'$ of $\Lambda_{C^{\perp}}$ of rank 2.
A generator matrix is
\[
{\bm B}' =
\begin{bmatrix}
{\bm a} \\
{\bm b}
\end{bmatrix}
=
\begin{bmatrix}
a_1 & \ldots & a_n \\
b_1 & \ldots & b_n
\end{bmatrix} {\bm B}
=
\begin{bmatrix}
a_1 & a_1+a_2q & \ldots &  a_1+a_{n-1}q & -a_1+a_{n}q \\
b_1 & b_1+b_2q & \ldots & b_1+b_{n-1}q &  -b_1+b_{n}q
\end{bmatrix}.
\]
For $n\geq 4$, set $a_1=b_3=1$ and $a_i=0$ for $i\neq 1$, $b_i=0$ for $i\neq 3$ to get
\[
{\bm B}' =
\begin{bmatrix}
1 & 1 & 1 & 1 &\ldots &  1 & -1 \\
0 & 0 & q & 0 & \ldots & 0 &  0
\end{bmatrix}.
\]
Then
\[
\det({\bm B}'({\bm B}')^T)=
\det
\begin{bmatrix}
n & q \\
q & q^2
\end{bmatrix}
= q^2(n-1).
\]
For $n=3$, take e.g.,
\[
{\bm B}' =
\begin{bmatrix}
1 & 1 & -1 \\
0 & 0 & q \\
\end{bmatrix},~\det({\bm B}'({\bm B}')^T)=2q^2.
\]
We prove next that
\[
d_2(\Lambda_{C^\perp})=\min(q^4,q^2(n-1)).
\]

By the Cauchy-Binet formula, we may decompose $\det({\bm B}'({\bm B}')^T)$ into ${n\choose 2}=\frac{n(n-1)}{2}$ terms, which we group into four sums, containing respectively $n-1$, $1$, $(n-2)\sum_{i=1}^{n-3}i=\tfrac{(n-3)(n-2)}{2}$ and $n-2$ terms, as follows:
\begin{eqnarray*}
\det(\Lambda') & = &
\det({\bm B}'({\bm B}')^T)\\
& = & \sum_{i=2}^{n-1}
\left(
\det
\begin{bmatrix}
a_1 & a_1+a_iq  \\
b_1 & b_1+b_iq
\end{bmatrix}
\right)^2
+
\left(
\det
\begin{bmatrix}
a_1 & -a_1+a_{n}q \\
b_1 & -b_1+b_{n}q
\end{bmatrix}\right)^2
\\
& + &
\sum_{j=2}^{n-2}
\sum_{i>j}^{n-1}
\left(
\det
\begin{bmatrix}
a_1 +a_jq  & a_1+a_iq  \\
b_1 +b_jq  & b_1+b_iq
\end{bmatrix}
\right)^2
+
\sum_{j=2}^{n-1}
\left(
\det
\begin{bmatrix}
a_1 +a_jq  & -a_1+a_nq  \\
b_1 +b_jq  & -b_1+b_nq
\end{bmatrix}
\right)^2\\
&=&
\sum_{i=2}^{n}
\left(
q
\det
\begin{bmatrix}
a_1  & a_i  \\
b_1  & b_i
\end{bmatrix}
\right)^2
\\
& + &
\sum_{j=2}^{n-2}
\sum_{i>j}^{n-1}
\left(q
\det
\begin{bmatrix}
a_1  & a_i  \\
b_1  & b_i
\end{bmatrix}
+
q\det
\begin{bmatrix}
a_j  & a_1  \\
b_j  & b_1
\end{bmatrix}
+
q^2\det
\begin{bmatrix}
a_j  & a_i  \\
b_j  & b_i
\end{bmatrix}
\right)^2
\\
&+&
\sum_{j=2}^{n-1}
\left(
q
\det
\begin{bmatrix}
a_1   & a_n  \\
b_1   & b_n
\end{bmatrix}
+
q
\det
\begin{bmatrix}
a_j  & -a_1 \\
b_j  & -b_1
\end{bmatrix}
+
q^2
\det
\begin{bmatrix}
a_j  & a_n  \\
b_j  & b_n
\end{bmatrix}
\right)^2\\
&=&
q^2 \left(
\sum_{i=2}^n D_{1i}^2
+ \sum_{j=2}^{n-2}
\sum_{i>j}^{n-1}
(D_{1i} + D_{j1} + q D_{ji})^2
+ \sum_{j=2}^{n-1}
(D_{1n} - D_{j1} + q D_{jn})^2
\right) \\
\end{eqnarray*}
where
\[
D_{ij}=
\det
\begin{bmatrix}
a_i  & a_j  \\
b_i  & b_j
\end{bmatrix}.
\]
Set
\[
A_n =\sum_{i=2}^n D_{1i}^2
+ \sum_{j=2}^{n-2}
\sum_{i>j}^{n-1}
(D_{1i} + D_{j1} + q D_{ji})^2
+ \sum_{j=2}^{n-1}
(D_{1n} - D_{j1} + q D_{jn})^2.
\]

We will prove by induction that $\min_{{\bm a}, {\bm b}} A_{n}=\min(q^2,n-1)$ for $n\geq 3$. More precisely, for a fixed $q$, when $n-2<q^2$,  we will prove that if $\min_{\bm a,\bm b} A_{n-1}=n-2$ then either $\min_{{\bm a}, {\bm b}}  A_{n}=n-1$ or $\min_{{\bm a}, {\bm b}}  A_{n}=q^2$. It is sufficient to prove that if $\min_{\bm a,\bm b} A_{n-1}=n-2$ then $A_{n}\geq n-1$ or $A_{n}\geq q^2$.

First, the case $n=2$ is immediate, and when $n=3$, we have $A_{3}=\min(q^2,2)=2$. Indeed
\[
A_3 = D_{12}^2+D_{13}^2 + (D_{13}-D_{21}+qD_{23})^2.
\]
If $A_3=1$, then exactly one square is $1$, and the other two are $0$. If $D_{12}=\pm 1$, it is not possible for $D_{13}$ to be $0$, and vice-versa. If $D_{13}-D_{21}+qD_{23}=\pm 1$, it is not possible to have $D_{12}=D_{13}=0$.

Now suppose $n\geq 4$ and the statement is true for $n-1$, that is $A_{n-1}\geq n-2$.

Since
$$
\bm B^{'}=\begin{bmatrix}
a_1 & a_1+a_2q & \ldots &  a_1+a_{n-1}q & -a_1+a_{n}q \\
b_1 & b_1+b_2q & \ldots & b_1+b_{n-1}q &  -b_1+b_{n}q
\end{bmatrix}
$$
we know that the case of length $n-1$ is obtained by deleting a column of index $i_{0}$ for $2\leq i_{0}\leq n-1$ from $\bm B^{'}$, so
we may choose an index $2\leq i_{0}\leq n-1$, for which we have
$$A_{n}=A_{n-1}+D_{1i_{0}}^{2}+\sum_{j=2}^{i_{0}-1}(D_{1i_{0}}+D_{j1}+qD_{ji_{0}})^{2}+\sum_{i=i_{0}+1}^{n-1}(D_{1i}+D_{i_{0}1}+qD_{i_{0}i})^{2}+(D_{1n}-D_{i_{0}1}+qD_{i_{0}n})^{2}.$$
Note that if $i_{0}=n-1$, the term $\sum_{i=i_{0}+1}^{n-1}(D_{1i}+D_{i_{0}1}+qD_{i_{0}i})^{2}$ disappears. Similarly if $i_0=2$, it is the term
$\sum_{j=2}^{i_{0}-1}(D_{1i_{0}}+D_{j1}+qD_{ji_{0}})^{2}$ which disappears.

Let $a_{i},b_{i}$, $1\leq i\leq n$, be any choice of integers which minimizes $A_{n}$ (while keeping ${\bm a},{\bm b}$ linearly independent). By induction, $A_{n-1}\geq n-2$. First, if this choice of  $a_{i},b_{i}$ leads to $A_{n-1}>n-2$, then the proof is done. Suppose it leads to $A_{n-1}=n-2$. If $A_{n}\geq n-1$ does not hold, then it must be that
$$D_{1i_{0}}^{2}+\sum_{j=2}^{i_{0}-1}(D_{1i_{0}}+D_{j1}+qD_{ji_{0}})^{2}+\sum_{i=i_{0}+1}^{n-1}(D_{1i}+D_{i_{0}1}+qD_{i_{0}i})^{2}+(D_{1n}-D_{i_{0}1}+qD_{i_{0}n})^{2}=0.$$
Since this is a sum of squares, for it to be $0$, we need all the squares to $0$. It is enough to show at least one square cannot be $0$ to complete the proof. If $D_{1i_0} \neq 0$, we are thus done. So suppose $D_{1i_0} = 0$ and we must have
$$\sum_{j=2}^{i_{0}-1}(D_{j1}+qD_{ji_{0}})^{2}+\sum_{i=i_{0}+1}^{n-1}(D_{1i}+qD_{i_{0}i})^{2}+(D_{1n}+qD_{i_{0}n})^{2}=0.$$

To have
$$D_{1i_{0}}=
\det
\begin{bmatrix}
a_1  & a_{i_0}  \\
b_1  & b_{i_0}
\end{bmatrix}=
0$$
implies the vector $(a_{1},b_{1})$ is a multiple, possibly 0, of
$ (a_{i_{0}},b_{i_{0}})$.  We write $(a_{1},b_{1})\sim(a_{i_{0}},b_{i_{0}})$ to denote the relation ``being a multiple of".  It is easy to see that this is an equivalence relation.

If there is at least one index $j$, $2\leq j\leq i_{0}-1$,  for which $D_{ji_{0}}\neq0$, then the term $(D_{j1}+qD_{ji_{0}})^{2}=0$ implies $D_{j1}$ is a multiple of $q$ where $j\neq i_{0}$. Then we have $A_{n-1}=n-2\geq D_{j1}^{2}\geq q^2$ which is a contradiction.  Thus we have $D_{ji_{0}}=0$ for all $2\leq j\leq i_{0}-1$ which implies $(a_{i_{0}},b_{i_{0}})\sim(a_{j},b_{j})$ for all $2\leq j\leq i_{0}-1$. For the same reason, we have $(a_{i_{0}},b_{i_{0}})\sim(a_{j},b_{j})$ for $i_{0}+1\leq j\leq n-1$ from the second term and $(a_{i_{0}},b_{i_{0}})\sim(a_{n},b_{n})$ from the last term. By transitivity,
$$
(a_{1},b_{1})\sim   (a_{i_0},b_{i_0}) \sim (a_{i},b_{i})$$ for all $2\leq i\leq n$ which is a contradiction, since in that case $A_{n}=A_{n-1}=0$.

We are left to treat the case $A_{n-1}=0$, that is $(a_{1},b_{1})\sim(a_{i},b_{i})$ for all $i\neq i_{0}$. In that case
$$A_{n}=D_{1i_{0}}^{2}+\sum_{j=2}^{i_{0}-1}(D_{1i_{0}}+qD_{ji_{0}})^{2}+\sum_{i=i_{0}+1}^{n-1}(D_{i_{0}1}+qD_{i_{0}i})^{2}+(-D_{i_{0}1}+qD_{i_{0}n})^{2}.$$
\begin{itemize}
\item If $D_{1i_0}=0$: $A_{n}=q^2(\sum_{j=2}^{i_{0}-1}D_{ji_{0}}^{2}+\sum_{i=i_{0}+1}^{n-1}D_{i_{0}i}^{2}+D_{i_{0}n}^{2})$. Since we need to keep $\bm a,\bm b$ linearly independent, we know that there is at least one $j_{0}$ such that $D_{i_{0}j_{0}}\neq0$ and $A_n\geq q^2$.
\item If $D_{1i_0}\neq 0$: If $D_{1i_{0}}$ is a multiple of $q$ then $A_{n}\geq q^2$. If $D_{1i_{0}}$ is not a multiple of $q$ then all terms $(D_{1i_{0}}+qD_{ji_{0}})^{2}\geq 1$ for $2\leq j\leq i_{0}-1$, and all terms $(D_{i_{0}1}+qD_{i_{0}i})^{2}\geq 1$ for $i_{0}+1\leq i\leq n-1$, and also the term $(-D_{i_{0}1}+qD_{i_{0}n})^{2}\geq 1$ which implies $A_{n}\geq n-1$.
\end{itemize}

This completes the proof.

\item
By Proposition \ref{prop:Dn}, we know that $d_{2}(\Lambda_{C})= 3$, and
by Proposition \ref{duiou},
\[
\gamma'_{n,2}(\Lambda_{C})
=\frac{1}{q^{2}}\sqrt{d_2(\Lambda_{C})d_2(\Lambda_{C^{\perp}})}
= \frac{1}{q^{2}}\sqrt{3 q^2\min(q^2,n-1)}.
\]
When $n=4$, there is no choice for the minimum, which always gives $n-1$ for all $q\geq 2$ and we get
\[
\gamma'_{4,2}(\Lambda_C) = \frac{3}{q},~
\gamma'_{4,2}(D_n) = \frac{3}{2},q=2.
\]
Consequently, when $n\geq 5$, we get $\gamma'_{n,2}\geq \sqrt{3}$.
\end{enumerate}
\end{proof}

\begin{Corollary}
We have
\[
 1.7321 \approx
\sqrt{3} \leq \gamma'_{5,2} \leq  2,~
1.7321 \approx
\sqrt{3} \leq
\gamma'_{7,2} \leq \frac{8}{3} \approx 2.6667
\]
and $\gamma_{n,2} \geq \sqrt{3}$ for $n\geq 5$ so
\[
1.7321 \leq \gamma_{5,2} \leq  2.
\]
\end{Corollary}
\begin{proof}
The lower bound follows from the above proposition, namely
$\gamma'_{n,2}\geq\sqrt{3}$.
The upper follows from the known inequality
$\gamma'_{n,2l}\leq(\gamma'_{n-l,l})^{2}$. For $n=5$, we have
\[
\gamma'_{5,2}\leq(\gamma'_{4,1})^2 =2
\]
while for $n=7$
\[
\gamma'_{7,2}\leq(\gamma_{6,1})^{2} \leq 8/3.
\]
We use Corollary \ref{cor:n57}  for the upper bounds on $\gamma_{5,2}$, and $\gamma'_{n,l}\leq \gamma_{n,l}$ for the lower bounds.
\end{proof}

We note that the above slightly improves the lower given by Corollary \ref{cor:n57}, namely:
\[
1.723 \approx \frac{3}{4^{2/5}} \leq \gamma_{5,2} \leq  2.
\]

\noindent {\bf Acknowledgment.} This work was supported by NSFC (Grant Nos. 12271199, 12441102) and China Scholarship Council (No. 202306770055). The hospitality of the division of mathematical sciences at Nanyang Technolgical University during the second author's visit is gratefully acknowledged.

\end{document}